\documentclass[preprint]{ptephy_v1}

\usepackage{hyperref}

\usepackage{mathrsfs}

\theoremstyle{plain}
\newtheorem{theorem}{Theorem}[section]
\newtheorem{lemma}[theorem]{Lemma}
\newtheorem{corollary}[theorem]{Corollary}
\newtheorem{proposition}[theorem]{Proposition}

\theoremstyle{definition}
\newtheorem{definition}[theorem]{Definition}
\newtheorem{example}[theorem]{Example}

\theoremstyle{remark}
\newtheorem{remark}{Remark}
\newtheorem{notation}{Notation}

\begin{document}

\title{Matrix Wigner Function and SU(1,1)}


\author{Peter Morrison}
\affil{University of Technology, Sydney \email{peter.morrison@uts.edu.au}}


\begin{abstract}%
This paper contains a brief sketch of some methods that can be used
to obtain the Wigner function for a number of systems. We give an
overview of the technique as it is applied to some simple differential
systems related to diffusion problems in one dimension. We compute
the Wigner function for the harmonic oscillator, the $xp$ interaction,
and a hyperbolic oscillator. These systems are shown to share several
properties in common related to the Whittaker function and various
formulae for the Laguerre polynomials. To contrast with the techniques
that are applicable to problems involving continuous states, we then
show that by expanding the solution space to the hyperbolic plane
and utilising some results from matrix calculus, we are able to recover
a number of interesting identities for SU(1,1) and the pseudosphere.
We close with a discussion of some more advanced topics in the theory
of the Wigner function.
\end{abstract}

\subjectindex{xxxx, xxx}

\maketitle

\begin{abstract}
This paper contains a brief sketch of some methods that can be used
to obtain the Wigner function for a number of systems. We give an
overview of the technique as it is applied to some simple differential
systems related to diffusion problems in one dimension. We compute
the Wigner function for the harmonic oscillator, the $xp$ interaction,
and a hyperbolic oscillator. These systems are shown to share several
properties in common related to the Whittaker function and various
formulae for the Laguerre polynomials. To contrast with the techniques
that are applicable to problems involving continuous states, we then
show that by expanding the solution space to the hyperbolic plane
and utilising some results from matrix calculus, we are able to recover
a number of interesting identities for SU(1,1) and the pseudosphere.
We close with a discussion of some more advanced topics in the theory
of the Wigner function.
\end{abstract}


\section{Introduction}

The quasiprobability function of Wigner \cite{wigner1997quantum} is an interesting
counterpart to the familiar probability distributions we associate
with the fundamental solutions of diffusion equations. This function
reduces, in many cases, to an analysis of the autocorrelation of a
stochastic variable as expressed through the eigenfunctions of the
solution. By using the Wigner function, we are able to extract all
expectations values in a similar way to the use of a transition probability
density. However, this function fails to be completely positive, and
obeys the laws of quasiprobabilities as opposed to probabilities as
a result. This paper addresses several deeper questions that might
be asked regarding the theory of probability densities in one dimension,
being the question of when it is possible to find a positive definite,
bounded solution to a differential system which defines a density.
As we shall show using this important, well-known counterexample,
it is possible to extract all measurable information from this quasiprobability,
and indeed, as we shall show in the following calculation, we can
learn valuable properties about the dynamics of diffusion problems
by relaxing the assumption of positivity.

In parallel with this investigation, we shall show that another assumption
related to the dimensionality of the system may be relaxed in a different
way, by utilising some relationships from the study of the quantum
brachistochrone \cite{morrison2012time}. The analysis proceeds in a similar
fashion, with the added complication of non-commutative variables.
We shall demonstrate that it is possible to use an identical methodology
as in the first part to obtain a number of interesting formulae by
using the theory of matrix representations.

\section{Review}

We shall briefly go over some papers and results that are used in
the following. As an entry point for the reader who is unfamiliar
with the Wigner function, the reference text of Tannor \cite{tannor2007introduction}
has an excellent discussion of the topic suitable for this level of
expertise. We refer other readers familiar with the basic structures
to the papers of \cite{frank2000wigner}, who has calculated the Wigner function
for the Morse oscillator. This system is very similar to those we
are to analyse in the following. For other references, especially
with regard to SU(1,1), the paper of Seyfarth et. al \cite{seyfarth2020wigner}
contains an outline of a calculation related to the Wigner function
on the hyperbolic disk, by using a coherent state methodology. The
papers of Lenz \cite{lenz2016mehler, lenz2017mehler, lenz2017positive} is also of use in understanding
the interaction between positive definite kernels and SU(1,1). This
paper is closely associated with some results explored in this work,
and although we shall not delve deeply into the use of coherent states
and their formalism, at many points the results shall be complementary.
We shall point this out where appropriate.

\section{Properties of the Wigner Transform}

We shall now briefly review the necessary formal axioms of Wigner
function theory and the projection theories relevant to this calculation.
In particular, we shall need the formulae for the Wigner function,
which is essentially the characteristic function of the autocorrelated
operator. This paper shall use the known relationship between the
phase-space version of quantum mechanics and matrix projection operators.
For references, the interested reader is directed to the historical
works of Wigner \cite{wigner1997quantum}, the canonical reference \cite{hillery1984distribution}
and description of the application of Poisson brackets and phase space
quantum mechanics by Groenewold \cite{wigner1997quantum, groenewold1946principles} and Moyal
\cite{moyal1949quantum}. More recent resources may be found in Tannor's treatise
on time dependent quantum mechanics \cite{tannor2007introduction} and the papers
of Curtright et. al \cite{curtright1998features, curtright2001generating}, the last of which
contains an explicit outline of the method of the associative star
product.

\subsection{Axioms of the Wigner Function}

We shall run over the basic properties of the Wigner transform, and
how it relates to the intersection between classical and quantum mechanics.
The basic definition of the Wigner function, see e.g. \cite{hillery1984distribution, tannor2007introduction}
for an modern introduction to the topic, can be written as the Fourier
transform of the autocorrelated operator expectation:
\begin{definition}
    
\begin{equation}
\mathcal{A}_{W}(p,q)=\int_{-\infty}^{+\infty}e^{ips/\hbar}\left\langle q-\dfrac{s}{2}\right|\hat{A}\left|q+\dfrac{s}{2}\right\rangle ds
\end{equation}
\end{definition}
The basic concept is to develop a way to move between operators and
continuous, differentiable functions. Wigner \cite{wigner1997quantum} developed
a series of axioms for his analysis of quasiprobability densities
by applying this functional to the density matrix of the system:
\begin{definition}
    
\begin{equation}
\mathcal{\rho}_{W}(p,q)=\int_{-\infty}^{+\infty}e^{ip(x-x')/\hbar}\left\langle x'\right|\hat{\varrho}\left|x\right\rangle ds
\end{equation}
\end{definition}

\begin{definition}
    The basic assumptions on the density matrix are that it be representated
as a pure state:
\begin{equation}
\hat{\varrho}=\left|\Psi\right\rangle \left\langle \Psi\right|
\end{equation}
\end{definition}

\begin{definition}
    We assume that the set of probabilities is exhaustive, so we may write:
\begin{equation}
\mathbf{Tr}\left(\hat{\varrho}\right)=\left\langle \Psi\right|\left.\Psi\right\rangle =1
\end{equation}
\end{definition}

\begin{definition}
  We can relate expectation values of an operator to a trace over the
Wigner function   
\begin{equation}
\left\langle \hat{A}\right\rangle =\mathbf{Tr}\left(\hat{\varrho}\hat{A}\right)=\dfrac{1}{2\pi\hbar}\int_{-\infty}^{+\infty}\int_{-\infty}^{+\infty}\mathcal{\rho}_{W}(p,q)\mathcal{A}_{W}(p,q).dp.dq
\end{equation}
\end{definition}
This property is representative of a broader class of operators which
are tracial in that the trace is preserved as an integral over the
phase space:
\begin{definition}
    
\begin{equation}
\mathbf{Tr}\left(\hat{A}\hat{B}\right)=\int_{-\infty}^{+\infty}\int_{-\infty}^{+\infty}\mathcal{A}_{W}(p,q)\mathcal{B}_{W}(p,q).dp.dq
\end{equation}
\end{definition}
Further, the Wigner quasiprobability function is defined by the normalised
quasidistribution: 
\begin{definition}
    
\begin{equation}
f_{W}(p,q)=\dfrac{1}{2\pi\hbar}\mathcal{\rho}_{W}(p,q)
\end{equation}
\end{definition}
Finally, as any probability density is a real number, this is associated
to the Hermitian nature of the Wigner transform. 
\begin{definition}
    
\begin{equation}
f(x,p)=\left\langle \Psi\left|\hat{M}(x,p)\right|\Psi\right\rangle 
\end{equation}
\begin{equation}
\hat{M}^{\dagger}=\hat{M}
\end{equation}
\begin{equation}
f^{*}=f
\end{equation}
\end{definition}
and therefore the Wigner function is real. Marginal distributions
are produced in the standard way 
\begin{lemma}
    
\begin{equation}
\int_{-\infty}^{+\infty}f(x,p)dp=|\psi(x)|^{2}=\left\langle x\left|\hat{\rho}\right|x\right\rangle 
\end{equation}
\begin{equation}
\int_{-\infty}^{+\infty}f(x,p)dx=|\psi(p)|^{2}=\left\langle p\left|\hat{\rho}\right|p\right\rangle 
\end{equation}
\end{lemma}
We assume also that the total probability sums to one, in the same
way as a classical probability density. The integral in this case
is taken over the whole phase space:
\begin{definition}
    
\begin{equation}
\int_{-\infty}^{+\infty}\int_{-\infty}^{+\infty}f(x,p)dxdp=1
\end{equation}
\end{definition}
The next ingredient is Galilean invariance principles. The basis propositions
state that the Wigner function should transform appropriately under
translations in space and time. For space translations, we have that
if the state is shifted in space, the Wigner function shifts accordingly
via:
\begin{equation}
\begin{array}{cc}
\Psi(x)\rightarrow & \Psi(x+a)\\
f(x,p)\rightarrow & f(x+a,p)
\end{array}
\end{equation}
and for momentum, we have that if the wavefunction is rotated through
a phase, the momentum component is translated according to:
\begin{equation}
\begin{array}{cc}
\Psi(x)\rightarrow & e^{ip'x/\hbar}\Psi(x)\\
f(x,p)\rightarrow & f(x,p-p')
\end{array}
\end{equation}
The final ingredient is parity invariance. If we invert the space
coordinate by mirror reflection, the Wigner function will change via:
\begin{equation}
\begin{array}{cc}
\Psi(x)\rightarrow & \Psi(-x)\\
f(x,p)\rightarrow & f(-x,-p)
\end{array}
\end{equation}
In a complementary way, taking the complex conjugate of the wave function,
then the momentum coordinate is altered to yield
\begin{equation}
\begin{array}{cc}
\Psi(x)\rightarrow & \Psi^{*}(x)\\
f(x,p)\rightarrow & f(x,-p)
\end{array}
\end{equation}

These are all the ingredients that we need to define the Wigner (c.
1932) distribution. Note that we have not assumed that the density
is positive. This is why it is called a quasiprobability. Many authors
have discussed the non-positivity of this function so we shall not
discuss it here.

\subsection{Poisson Bracket}

We shall now review some basic properties of the Poisson bracket and
elementary phase space mechanics.
\begin{definition}
     The Poisson bracket gives the Liouville
evolution for the state space: 
\begin{equation}
\dfrac{\partial\rho}{\partial t}=\{\mathcal{H}(q,p),\rho\}
\end{equation}

The Poisson bracket is antisymmetric: 
\begin{equation}
\{a,b\}=-\{b,a\}
\end{equation}
and obeys the Jacobi identity: 
\begin{equation}
\{a,\{b,c\}\}+\{b,\{c,a\}\}+\{c,\{a,b\}\}=0
\end{equation}
\end{definition}
\begin{theorem}
    
The basic system dynamics is described through the action of the 2-form:
\begin{equation}
\omega=dp_{\mu}\wedge dq^{\mu}
\end{equation}
The generator of the Lie vector field is defined through the action
of the Poisson bracket: 
\begin{equation}
\hat{X}_{f}=\dfrac{\partial f}{\partial p_{\mu}}\dfrac{\partial}{\partial q^{\mu}}-\dfrac{\partial f}{\partial q^{\mu}}\dfrac{\partial}{\partial p_{\mu}}
\end{equation}

Hamilton's equations of classical dynamics can then be written as:
\begin{equation}
f=\mathcal{H}(q^{\mu},p_{\mu})
\end{equation}
\begin{equation}
\dot{q}^{\mu}=\dfrac{\partial\mathcal{H}}{\partial p_{\mu}}
\end{equation}
\begin{equation}
\dot{p}_{\mu}=-\dfrac{\partial\mathcal{H}}{\partial q^{\mu}}
\end{equation}
\begin{equation}
\hat{X}_{\mathcal{H}}=\dot{q}^{\mu}\dfrac{\partial}{\partial q^{\mu}}+\dot{p}_{\mu}\dfrac{\partial}{\partial p_{\mu}}=\dfrac{\partial}{\partial t}
\end{equation}
\end{theorem}
We shall now describe the basic equivalence from the phase space picture
to quantum dynamics. From quantum mechanics, we know the equivalent
of the Liouville equation in statistical theory is the von Neumann
equation: 
\begin{equation}
i\dfrac{\partial\hat{\rho}}{\partial t}=[\hat{H},\hat{\rho}]
\end{equation}
where under the bracket we have non-commutative matrix operations
via: 
\begin{equation}
[\hat{A},\hat{B}]=\hat{A}\hat{B}-\hat{B}\hat{A}
\end{equation}
The dynamic operator is Hermitian: 
\begin{equation}
\hat{H}^{\dagger}=\hat{H}
\end{equation}
We assume the standard properties of a pure state density matrix:
\begin{equation}
\mathrm{Tr}\hat{\rho}=1,\hat{\rho}^{2}=\hat{\rho}
\end{equation}

This is the basic setup of the properties of the Wigner function.
There are various situations more general than this, in particular
it is clear that this system is Galilei-invariant and not Lorentz
invariant, which limits the domain of application of this theory i.e.
it is not relativistic. In the case where velocities are low this
is an acceptable situation, and we can expect there to be a systematic
relationship between the behaviour of the quantum and classical theories.

\subsection{Moyal Bracket}

Groenewold (1946), J.E. Moyal (1949) \cite{groenewold1946principles, moyal1949quantum} proposed
to investigate the nature of phase-space quantum mechanics. Their
insight was to try and understand the exact relationship between the
classical and quantum systems by using the Poisson bracket. Independently,
they came to similar answers regarding the question as to what replaces
the Poisson bracket in quantum mechanics. The solution to this question
was to employ the associative star product, defined through:
\begin{definition}
    
\begin{equation}
'\star'=\exp\left(\dfrac{i\hbar}{2}\left[\overleftarrow{\partial}_{x}\overrightarrow{\partial}_{p}-\overleftarrow{\partial}_{p}\overrightarrow{\partial}_{x}\right]\right)=\exp\left(-\dfrac{i\hbar}{2}\hat{\Lambda}\right)
\end{equation}
\end{definition}
\begin{notation}
    
Overhead arrows in this context indicate if the operator acts to the
left or the right:
\begin{equation}
f\hat{\Lambda}g=-\left(\dfrac{\partial f}{\partial x}\dfrac{\partial g}{\partial p}-\dfrac{\partial f}{\partial p}\dfrac{\partial g}{\partial x}\right)
\end{equation}
\end{notation}

The great insight of Moyal was to employ this algebra to derive the
quantum equivalent of phase space mechanics. Indeed, by using this
method we can solve for the dynamics of the Wigner distribution defined
through the autocorrelation transform. The Moyal equation may be written:
\begin{equation}
W_{\hat{A}\hat{B}}=W_{\hat{A}}\exp\left(-\dfrac{i\hbar}{2}\hat{\Lambda}\right)W_{\hat{B}}
\end{equation}
where $W_{\hat{A}}$ represents the Wigner transform of the operator.
The quantum equivalent of the Liouville equation in phase space mechanics
goes over into Moyal's equation: 
\begin{equation}
\dfrac{\partial}{\partial t}W_{\hat{\rho}}=-\dfrac{2}{\hbar}W_{\hat{H}}\sin\left(\dfrac{\hbar}{2}\hat{\Lambda}\right)W_{\hat{\rho}}
\end{equation}
The classical limit of this equation can be evaluated which gives
\begin{equation}
\dfrac{\partial}{\partial t}W_{\hat{\rho}}=-W_{\hat{H}}\hat{\Lambda}W_{\hat{\rho}}=\{W_{\hat{H}},W_{\hat{\rho}}\}
\end{equation}
The normal state of affairs in quantum mechanics is the eigenvalue
problem: 
\begin{equation}
\mathcal{H}\psi=E\psi
\end{equation}
In terms of the star-eigenvalue problem, it is simple to show that
there is a complementary equation: 
\begin{equation}
\mathcal{H}\star f=Ef
\end{equation}
The action of the associative star product is shown to be: 
\begin{equation}
f\star g=f(x+\dfrac{i\hbar}{2}\overrightarrow{\partial}_{p},p-\dfrac{i\hbar}{2}\overrightarrow{\partial}_{x})g(x,p)
\end{equation}
and the star-eigenvalue problem becomes the differential equation:
\begin{theorem}
    
\begin{equation}
\mathcal{H}(x+\dfrac{i\hbar}{2}\overrightarrow{\partial}_{p},p-\dfrac{i\hbar}{2}\overrightarrow{\partial}_{x})f(x,p)=Ef(x,p)
\end{equation}
\end{theorem}
This is actually two simultaneous differential equations, in that
the solution to this expression solves a complex valued PDE which
we can separate into real and imaginary parts. For description of
this technique in a modern context, consult \cite{curtright1998features, curtright2001generating}. 

\section{Wigner Distribution for Hyperbolic Systems}

The star eigenvalue equations guarantee ready access to the formulae
for the Wigner distribution, if we can solve them. As we shall show,
it is possible to use the Wigner transform to find relationships between
the eigenstates of different systems. We shall also encounter the
spectral theory of the Wigner distribution, and demonstrate that we
can use this to find kernel representations for the various different
systems. At its core, the simplified example of a hyperbolic geometry
in phase space gives us access to several fundamental examples which
we can use to probe the technique. Future works to appear have shown
that the Wigner-Weyl transform can be used to find Wigner distributions
in systems which have completely continuous eigenvalues, where the
basic theory of Curtright \cite{curtright1998features, curtright2001generating} suffers from
complications which do not present themselves for the discrete states
we shall consider in the following.

\subsection{Simple Harmonic Oscillator}
\begin{proposition}
    It is possible to use the star-eigenvalue equation to define Wigner functions consistent with the results of Groenwold and Moyal, (Curtright).
\end{proposition}
\begin{proposition}
    The Wigner function in this case is given by a Laguerre function of the distance function in phase space.
\end{proposition}

\begin{proof}
    There are very few results available for calculated Wigner functions.
We shall begin with the simplest non-trivial example, which is given
by the harmonic oscillator. This has the advantage of having a known
solution. The basic differential operator we shall take is specified
by: 
\begin{example}
    
\begin{equation}
\hat{H}=\dfrac{\hat{p}^{2}}{2m}+\dfrac{1}{2}m\omega^{2}x^{2}\sim\dfrac{1}{2}(\hat{p}^{2}+\hat{x}^{2})
\end{equation}
\end{example}
where we are using natural units to simplify the argument. The eigenfunctions
are given by the normalised Hermite polynomials: 
\begin{equation}
\Psi_{m}(u)=\dfrac{\pi^{-1/4}}{\sqrt{2^{m}m!}}H_{m}(u)e^{-u^{2}/2}
\end{equation}
The star-eigenvalue equation in this situation can be written: 
\begin{lemma}
    
\begin{equation}
\dfrac{1}{2}[(p-\dfrac{i}{2}\overrightarrow{\partial}_{x})^{2}+(x+\dfrac{i}{2}\overrightarrow{\partial}_{p})^{2}]f(x,p)=Ef(x,p)
\end{equation}
\end{lemma}

We use normalised variables: 
\begin{equation}
\omega=\hbar=m=c=1
\end{equation}
Taking real and imaginary parts, we find the differential equations:
\begin{equation}
x\dfrac{\partial f}{\partial p}-p\dfrac{\partial f}{\partial x}=0
\end{equation}
\begin{equation}
\left[(p^{2}+x^{2})-\dfrac{1}{4}(\dfrac{\partial^{2}}{\partial p^{2}}+\dfrac{\partial^{2}}{\partial x^{2}})-2E\right]f(x,p)=0
\end{equation}
The two differential equations are solved using the following. 
\begin{equation}
f(x,p)=F(\lambda[x^{2}+p^{2}])
\end{equation}
\begin{equation}
f(x,p)=F(z),z=2(x^{2}+p^{2})
\end{equation}
Substituting this into the second differential equation and using
the chain rule, we find: 
\begin{equation}
\left(\dfrac{z}{4}-\dfrac{\partial}{\partial z}-z\dfrac{\partial^{2}}{\partial z^{2}}-E\right)F(z)=0
\end{equation}
Using the trial solution $F(z)=e^{-z/2}L(z)$ this reduces to the
Laguerre differential equation 
\begin{equation}
(E-\dfrac{1}{2})L(z)+(1-z)\dfrac{\partial L}{\partial z}+z\dfrac{\partial^{2}L}{\partial z^{2}}=0
\end{equation}
The system is solved using the set of eigenstates $n=E-\dfrac{1}{2}=0,1,2,...$,
then we have $L(z)=L_{n}(z)$, and the solution for the Wigner distribution
is then given by: 
\begin{equation}
f_{n}(x,p)=C_{n}F(z)=C_{n}e^{-z/2}L_{n}(z)=C_{n}e^{-2\mathcal{H}}L_{n}(4\mathcal{H})
\end{equation}
or in the original coordinates: 
\begin{theorem}
    
\begin{equation}
f_{n}(x,p)=C_{n}e^{-(x^{2}+p^{2})}L_{n}[2(x^{2}+p^{2})]
\end{equation}
\end{theorem}
To resolve the value of the constant, we must use other means. The
most straightforward way in which to determine this is to use the
integral normalisation of the quasidistribution. The integral formula
for the Wigner distribution can be written out as: 
\begin{equation}
f_{m}(x,p)=\int_{-\infty}^{+\infty}e^{2ipy}\Psi_{m}^{*}(x+y)\Psi_{m}(x-y)dy
\end{equation}
Using the expression for the eigenstates of the harmonic oscillator,
and doing some simple calculations involving Hermite polynomials yield:
\begin{equation}
f_{m}(x,p)=\dfrac{1}{2^{m}m!\sqrt{\pi}}\int_{-\infty}^{+\infty}e^{2ipy}H_{m}(x+y)H_{m}(x-y)e^{-(x^{2}+y^{2})}dy
\end{equation}
This equation will be further simplified by using an integral formula.
Extracting factors which do not enter into the integration, the formula
may be manipulated into the format: 
\begin{equation}
f_{m}(x,p)=\dfrac{(-1)^{m}e^{-(x^{2}+p^{2})}}{2^{m}m!\sqrt{\pi}}\int_{-\infty}^{+\infty}e^{-(y-ip)^{2}}H_{m}(y+x)H_{m}(y-x)dy
\end{equation}
Using formula from Gradshteyn and Ryzik (eq. 7.377 pp 853)\cite{gradshteyn2014table}:
\begin{lemma}
    
\begin{equation}
L_{m}(-2\xi_{1}\xi_{2})=\dfrac{1}{2^{m}\Gamma(m)\sqrt{\pi}}\int_{-\infty}^{+\infty}e^{-\eta^{2}}H_{m}(\eta+\xi_{1})H_{m}(\eta+\xi_{2})d\eta
\end{equation}
\begin{equation}
\xi_{1}=ip+x,\xi_{2}=ip-x,\xi_{1}\xi_{2}=-(x^{2}+p^{2})
\end{equation}
\end{lemma}
The solution for the Wigner distribution from this perspective is
then: 
\begin{equation}
f_{m}(x,p)=\dfrac{(-1)^{m}e^{-(x^{2}+p^{2})}}{\pi}L_{m}[2(x^{2}+p^{2})]
\end{equation}
The only major difference between the results is the normalisation.
This can be evaluated using the statement of total probability $\int_{-\infty}^{+\infty}\int_{-\infty}^{+\infty}f_{n}(x,p)dxdp=1$.
Using this on the un-normalised form of the Wigner distribution, we
find: 
\begin{equation}
C_{n}\int_{-\infty}^{+\infty}\int_{-\infty}^{+\infty}e^{-(x^{2}+p^{2})}L_{n}[2(x^{2}+p^{2})]dxdp=1
\end{equation}
In the space represented by the co-ordinate $z=2(x^{2}+p^{2})$, the
integral over all of phase space is transformed via 
\begin{equation}
\int_{-\infty}^{+\infty}\int_{-\infty}^{+\infty}dxdp\rightarrow\dfrac{1}{2}\int_{-\infty}^{+\infty}dz
\end{equation}
Normalisation of the quasidistribution is then: 
\begin{equation}
\dfrac{C_{n}}{2}\int_{0}^{+\infty}dze^{-z/2}L_{n}(z)=1
\end{equation}
Using Gradshetyn and Ryzik (eq. 7.414.6 pp809) \cite{gradshteyn2014table}:
\begin{equation}
\int_{0}^{+\infty}e^{-bx}L_{n}(x)dx=(b-1)^{n}b^{-n-1}
\end{equation}
we can resolve the normalised Wigner function to be given by: 
\begin{theorem}
    
\begin{equation}
f_{n}(x,p)=(-1)^{n}e^{-(x^{2}+p^{2})}L_{n}[2(x^{2}+p^{2})]
\end{equation}
\end{theorem}
in agreement with the other methods. 

\end{proof}
\begin{remark}
The missing numerical factor
of $\pi$ comes from our original definition of the distribution itself.
This solution dates to \cite{moyal1949quantum}. Obviously the class of systems
which are deformable into a simple harmonic oscillator and the myriad
of variations forms a large part of quantum mechanics. Consequently,
this representation of the Wigner function performs an analogous role
to the kernel function in the two point correlation function.

\end{remark}

\subsection{XP oscillator}

We are led naturally to think of other examples related to the harmonic
oscillator. There are several directions one may take, the simplest
being the change from spherical symmetry, as expressed through the
simple harmonic oscillator, to a hyperbolic geometry. 

\begin{proposition}
    The XP oscillator is a hyperbolic variant of the simple harmonic oscillator.
\end{proposition}
\begin{proposition}
    The Wigner function in this case can be evaluated by similar means, with the distance metric modified to a suitable form for the hyperbolic space, and the Laguerre polynomial replaced by a Whittaker function.
\end{proposition}

\begin{proof}

The following
two systems are indicative of various ways in which this change can
be achieved satisfactorily, the hyperbolic plane in multiplicative
form: 
\begin{definition}
    
\begin{equation}
\mathcal{H}=\dfrac{1}{2}(\hat{x}\hat{p}+\hat{p}\hat{x})=-i\left(x\dfrac{\partial}{\partial x}+\dfrac{1}{2}\right)
\end{equation}
\end{definition}
and the hyperbolic oscillator which we can write as: 
\begin{definition}
    
\begin{equation}
\mathcal{H}=\dfrac{1}{2}(p^{2}-x^{2})=-\dfrac{1}{2}\left[\dfrac{\partial^{2}}{\partial x^{2}}+x^{2}\right]
\end{equation}
\end{definition}
We assume the standard form of the momentum operator $\hat{p}=-i\partial_{x}$
familiar from quantum mechanics. The three fundamental operators defined
by these different Hamiltonians + the SHO are related to the fundamental
invariants that can be formed between the space and momentum. We have
a method for evaluating the Wigner function in a unitary space, but
the formulae for these Hamiltonian operators is only pseudounitary,
hence our basic method requires modification to function correctly
in this case. For the situation where the dynamics is specified through:
\begin{equation}
\mathcal{H}f=-i\left(x\dfrac{\partial}{\partial x}+\dfrac{1}{2}\right)f
\end{equation}
we can carry out most of the analysis in exactly the same way as for
the harmonic oscillator. In this case, the star-eigenvalue equations
are given by: 
\begin{equation}
\mathcal{H}(x+\dfrac{i\hbar}{2}\overrightarrow{\partial}_{p},p-\dfrac{i\hbar}{2}\overrightarrow{\partial}_{x})f(x,p)=Ef(x,p)
\end{equation}
and can be written: 
\begin{theorem}
    
\begin{equation}
xpf+\dfrac{1}{4}\dfrac{\partial^{2}f}{\partial p\partial x}+\dfrac{i}{2}\left(p\dfrac{\partial f}{\partial p}-x\dfrac{\partial f}{\partial x}\right)=Ef
\end{equation}
\end{theorem}
Taking real and imaginary parts, we obtain a solvable set of equations.
Writing out the paired equations explicitly, we find: 
\begin{theorem}
    
\begin{equation}
\left(xp+\dfrac{1}{4}\dfrac{\partial^{2}}{\partial p\partial x}\right)f=Ef
\end{equation}
\begin{equation}
p\dfrac{\partial f}{\partial p}-x\dfrac{\partial f}{\partial x}=0
\end{equation}
\end{theorem}
Similar considerations to the harmonic oscillator calculation allow
us to derive the following: 
\begin{equation}
F(z)=C_{1}W_{iE,1/2}(4iz)+C_{2}M_{iE,1/2}(4iz)
\end{equation}
where the variable changes are given by $f(x,p)=F(xp)=F(z)$ . To
analyse which of these solutions has the correct behaviour, we must
examine the structure of the Wigner distribution. Proceeding in a
similar way to the SHO calculation, we can write the integral form
of the Wigner distribution as: 
\[
f(x,p)=\int e^{ips}\Psi_{m}^{*}(x-s/2)\Psi_{m}(x+s/2)ds
\]
\begin{equation}
=|d_{m}|^{2}\int_{0}^{\infty}e^{ips}(x+\dfrac{s}{2})^{-iE_{m}}(x-\dfrac{s}{2})^{iE_{m}}ds
\end{equation}
We have assumed the eigenfunction for the pseudo-Hamiltonian operator
can be specified by $\Psi_{m}(x)=d_{m}(x)x^{iE_{m}}$. Evaluating
the integral can be done using Whittaker functions: 
\begin{theorem}
    
\begin{equation}
W_{iE_{m},1/2}(4ixp)=\dfrac{e^{-2ixp}}{\Gamma(1-iE_{m})}\int_{0}^{\infty}e^{-4ixpu}u^{-iE_{m}}(1+u)^{iE_{m}}du
\end{equation}
\end{theorem}
after some lengthy algebra. This form of the confluent hypergeometric
function is new and is not found in \cite{buchholz2013confluent}, but follows
naturally from the geometry of the problem. The result found for the
Wigner function in this instance is given by the un-normalised formula:
\begin{equation}
f_{m}(x,p)=C_{m}(-1)^{-iE_{m}}W_{iE_{m},1/2}(4ixp)
\end{equation}
whence upon using the statement of total probability we conclude that
the normalisation is given by: 
\begin{equation}
\int AW_{iE_{m},1/2}(4ixp)dxdp=1
\end{equation}
A formula from Gradshetyn and Ryzik (eq 7.622.11) \cite{gradshteyn2014table}
enables resolution of the constant: 
\begin{equation}
\int_{0}^{\infty}e^{-x/2}x^{\nu-1}W_{\kappa,1/2}(x)dx=\dfrac{\Gamma(\nu)\Gamma(\nu+1)}{\Gamma(\nu-\kappa+1)}
\end{equation}
where the integration measure in this case goes over $\int\int dxdp=\int_{0}^{\infty}dz$.
Putting all the ingredients together, we find the normalised formula
for the Wigner distribution on the half-plane as given by: 
\begin{theorem}
    
\begin{equation}
f_{m}(x,p)=(-1)^{-iE_{m}}\dfrac{e^{-2ixp}}{(4ixp)}\Gamma(1-iE_{m})W_{iE_{m},1/2}(4ixp)
\end{equation}
\end{theorem}    
\end{proof}
\begin{remark}
    This is a new result. Note the appearance of the ``distance'' function
in the multiplicative hyperbolic system in the arguments of the different
functions. We shall now show that the other example can be addressed
in a similar fashion.

\end{remark}
\subsection{Hyperbolic Oscillator}
\begin{proposition}
    A third variant with an analytic Wigner function is given by the hyperbolic oscillator.
\end{proposition}
\begin{corollary}
    This can be viewed as the extension of the harmonic oscillator to complex momentum or space.
\end{corollary}
\begin{proposition}
    The Wigner function may be found in a similar fashion as with the previous two examples. In this case, the distance metric is replaced by a hyperbolic distance in phase space, and the Wigner function is given by a special type of Laguerre function.
\end{proposition}
\begin{proof}

The hyperbolic oscillator equation is defined by the Hamiltonian:
\begin{definition}
    
\begin{equation}
\mathcal{H}=\dfrac{1}{2}(p^{2}-x^{2})
\end{equation}
\end{definition}
We can think of the three systems as being either spherical, pseudospherical,
or hyperspherical. We can see directly here that this system will
not be positive definite. There are deeper connections to differential
geometry here which may be seen as deformations of the Hamiltonian
into a curved space. See e.g. \cite{grosche1988path, grosche1990path} for an outline
of the path integral kernel on the hyperbolic plane/pseudosphere.
Calculating the Wigner function here is much harder, as instead of
unitary, or pseudounitary, we have to deal with a different type of
symmetry relation. Writing out the star-eigenvalue equations in full,
we have:
\begin{theorem}
    
\begin{equation}
\mathcal{H}\star f=(p^{2}-x^{2})f+\dfrac{1}{4}\left(\dfrac{\partial^{2}f}{\partial p^{2}}-\dfrac{\partial^{2}f}{\partial x^{2}}\right)-i\left(p\dfrac{\partial f}{\partial x}+x\dfrac{\partial f}{\partial p}\right)=Ef
\end{equation}
The fundamental differential equation for the eigenstate is given
by:
\begin{equation}
\mathcal{H}\Psi_{m}(x)=E_{m}\Psi_{m}(x)
\end{equation}
\begin{equation}
\mathcal{H}=-\dfrac{1}{2}\left[\dfrac{\partial^{2}}{\partial x^{2}}+x^{2}\right]
\end{equation}
\end{theorem}
Now, the star-eigenvalue equations can be solved to find: 
\begin{lemma}
    
\begin{equation}
f(x,p)=F(p^{2}-x^{2})=F(z)
\end{equation}
\begin{equation}
\dfrac{\partial^{2}F}{\partial z^{2}}+(1-\dfrac{E}{z})F(z)=0
\end{equation}
\begin{equation}
F(z)=C_{1}M_{iE/2,1/2}(2iz)+C_{2}W_{iE/2,1/2}(2iz)
\end{equation}
\end{lemma}
There are two cases which can be distinguished here, which could take
the place of the Hermitian symmetry. The eigenvalue might be a complex
number, in which case: 
\begin{equation}
-\dfrac{i}{2}\left[\dfrac{\partial^{2}\Psi_{m}}{\partial x^{2}}+x^{2}\Psi_{m}\right]=E_{m}\Psi_{m}
\end{equation}
with adjoint 
\begin{equation}
\dfrac{i}{2}\left[\dfrac{\partial^{2}\bar{\Psi}{}_{m}}{\partial x^{2}}+x^{2}\bar{\Psi}_{m}\right]=E_{m}\bar{\Psi}_{m}
\end{equation}
The solutions are easily found: 
\begin{equation}
\Psi_{m}=\dfrac{1}{\sqrt{x}}\left[c_{m}^{+}M_{-E/2,1/4}(ix^{2})+c_{m}^{-}W_{-E/2,1/4}(ix^{2})\right]
\end{equation}
\begin{equation}
\bar{\Psi}_{m}=\dfrac{1}{\sqrt{x}}\left[c_{m}^{+}M_{+E/2,1/4}(ix^{2})+c_{m}^{-}W_{+E/2,1/4}(ix^{2})\right]
\end{equation}
The other alternative is that the eigenvalue is a real number: 
\begin{equation}
\mathcal{H}\Psi_{m}=-\dfrac{1}{2}\left[\dfrac{\partial^{2}\Psi_{m}}{\partial x^{2}}+x^{2}\Psi_{m}\right]=E_{m}\Psi_{m}
\end{equation}
and the adjoint is skew Hermitian. 
\begin{equation}
\bar{\Psi}_{m}\mathcal{H}^{\dagger}=+\dfrac{1}{2}\left[\dfrac{\partial^{2}\bar{\Psi}_{m}}{\partial x^{2}}+x^{2}\bar{\Psi}_{m}\right]=\bar{\Psi}_{m}E_{m}
\end{equation}
The solutions in this case can also be written in terms of Whittaker
functions: 
\begin{lemma}
    
\begin{equation}
\Psi_{m}=\dfrac{1}{\sqrt{x}}\left[c_{m}^{+}M_{iE_{m}/2,1/4}(ix^{2})+c_{m}^{-}W_{iE_{m}/2,1/4}(ix^{2})\right]=\Psi_{m}(x,E)
\end{equation}
\end{lemma}
We have a symmetry relation $\bar{\Psi}_{m}=\Psi_{m}(x,-E)$ for the
adjoint. This represents the correct solution. The correct part of
the solution taking into account boundary conditions is: 
\begin{equation}
\Psi_{m}(x)=\dfrac{c_{m}W_{iE_{m}/2,1/4}(ix^{2})}{\sqrt{x}}
\end{equation}
The adjoint in this case is: 
\begin{equation}
\bar{\Psi}_{m}(x)=\dfrac{c_{m}W_{-iE_{m}/2,1/4}(ix^{2})}{\sqrt{x}}
\end{equation}
Evaluating the Wigner function in this situation seems to be a daunting
task, however we can access the results by using the theory of special
functions and Hermite polynomials. The following known conversion
formulae are available for Whittaker functions, Laguerre functions
and the Hermite polynomials. For the even states, we may write:
\begin{theorem}
    
\[
\dfrac{W_{n+1/4,-1/4}(z)}{\sqrt{z}}=(-1)^{n}\Gamma(n-1/2)z^{-1/4}L_{n-1/2}^{(-1/2)}(z)e^{-z/2}
\]
\begin{equation}
=z^{1/4}2^{-2n}H_{2n}(\sqrt{z})e^{-z/2}
\end{equation}
and the complementary formulae for the odd states are given by:
\[
\dfrac{W_{n+3/4,1/4}(z)}{\sqrt{z}}=(-1)^{n}\Gamma(n)z^{1/4}L_{n}^{(1/2)}(z)
\]
\begin{equation}
=z^{1/4}2^{-2n}H_{2n+1}(\sqrt{z})e^{-z/2}
\end{equation}
\end{theorem}
It is important to realise that this is a physical model, and as such
we must respect the difference in parity between these two sets of
complementary eigenstates which form the basis of the system. The
other key ingredient here is the order switching formula for Laguerre
polynomials: 
\begin{equation}
\dfrac{(-x)^{k}}{k!}L_{n}^{(k-n)}(x)=\dfrac{(-x)^{n}}{n!}L_{n}^{-(k-n)}(x)
\end{equation}
Another valuable formula can be found in Gradshetyn and Ryzik 7.377
\cite{gradshteyn2014table}:
\begin{equation}
\int_{-\infty}^{+\infty}e^{-u^{2}}H_{m}(u+y)H_{n}(u+z)du=2^{n}\sqrt{\pi}\Gamma(m)z^{n-m}L_{m}^{(n-m)}(-2yz)
\end{equation}
We are now in a position to evaluate the Wigner distribution for this
system:
\begin{equation}
f_{m}(x,p)=\int_{-\infty}^{+\infty}e^{2ipy}\bar{\Psi}_{m}(x-iy)\Psi_{m}(x+iy)dy
\end{equation}
where we note that the formula for the Wigner distribution in this
system is not the same due to the imaginary number in the argument
of the eigenstates. This is the analytical continuation of the autocorrelation.
Using the properties of Hermite polynomials, the solution is shown
to be: 
\begin{theorem}
    
\begin{equation}
f_{n}(x,p)=iAe^{\zeta}(-1)^{2n}2^{2n+1}\sqrt{\pi}\Gamma(2n)L_{2n}^{(1)}[-2\zeta]
\end{equation}
\begin{equation}
\zeta=p^{2}-x^{2}
\end{equation}
\end{theorem}
Comparing this with the solution from the star-eigenvalue equation,
we have: 
\begin{equation}
f_{m}(x,p)=CW_{iE_{m}/2,1/2}(2iz)
\end{equation}
There is a way to relate these two seemingly different representations.
The conversion formula gives: 
\begin{equation}
W_{m+1,1/2}(z)=(-1)^{m}\Gamma(m)e^{-z/2}zL_{m}^{(1)}(z)
\end{equation}
Q.E.D.
\end{proof}
\begin{remark}
    
Using the expressions for the Wigner function, it is simple to see
that $f_{n}(x,p)=F_{n}(\zeta)$ and the other form of the Wigner state
is $F_{n/2}(2iz)$. We have not considered the statement of total
probability here.

\end{remark}
\section{Matrix Wigner Function}

We now examine finite matrix groups that have similar properties to
the functions we generated in the previous calculations. Some results
are available, but in general research is only thinly available on
the topic. Recent advances may be found in Seyfarth et. al \cite{seyfarth2020wigner}.
One way in which to generate matrix groups with an equivalent structure
to the previous sets of special functions is to use the creation and
annihilation representation of the coherent state. The matrix calculus
is then determined by the displacement and squeezing of the distribution.
We shall give a brief summary of known results. The aim is to determine
the basic determining equations which can be used to characterise
a Wigner function via the continuous differential operators, the final
part of this paper shall then show that there exist similar types
of relationships for particular special matrices. Briefly, we hope to address the following:
\begin{proposition}
    Is there a matrix equivalent for the Wigner function on SU(1,1)?
\end{proposition}
\begin{proposition}
    What is the correct way in which to model the dynamic evolution of the state space?
\end{proposition}

\subsection{Displacement Operator}
\begin{proposition}
    The displacement operator defines the phase space version of quantum mechanics.
\end{proposition}
\begin{corollary}
    The Stone-von Neumann theorem emerges as a result of braiding relationships.
\end{corollary}
\begin{proposition}
    The group representation theory of phase space quantum mechanics can be generated through the application of the displacement operator.
\end{proposition}
\begin{proof}

In the hyperbolic plane, one important operator is that of displacement.
This operator is a Weyl transform in the creation and annihilation
operators as shown in \cite{potovcek2015exponential}: 
\begin{definition}
    
\begin{equation}
\hat{D}(\alpha)=\exp\left(\alpha\hat{a}^{\dagger}-\alpha^{*}\hat{a}\right)
\end{equation}
\end{definition}
The composition formula is:
\begin{lemma}
    
\begin{equation}
\hat{D}(\alpha)\hat{D}(\alpha')=\exp\left(\dfrac{1}{2}(\alpha\alpha'^{*}-\alpha^{*}\alpha')\right)\hat{D}(\alpha+\alpha')
\end{equation}
\end{lemma}
In terms of the position and momentum coordinates it is simple to
show such laws as: 
\begin{lemma}
    
\begin{equation}
\hat{D}(x,p)\hat{D}(x',p')=\exp\left(i(px'-xp')/2\right)\hat{D}(x+x',p+p')
\end{equation}
\end{lemma}
Other identities include the braiding relationship: 
\begin{theorem}
    
\begin{equation}
\hat{D}(\alpha)\hat{D}(\beta)=e^{(\alpha\beta^{*}-\beta\alpha^{*})/2}\hat{D}(\beta)\hat{D}(\alpha)
\end{equation}
\end{theorem}
which are easily derived using some elementary applications of the
commutation rules for creation and annihilation operators. We assume
as always the standard boson algebra: 
\begin{definition}
    
\begin{equation}
\hat{a}=\dfrac{1}{\sqrt{2}}\left(\hat{x}+i\hat{p}\right),\hat{a}^{\dagger}=\dfrac{1}{\sqrt{2}}\left(\hat{x}-i\hat{p}\right)
\end{equation}
\begin{equation}
\left[\hat{a},\hat{a}^{\dagger}\right]=1,\left[\hat{x},\hat{p}\right]=i1
\end{equation}
\end{definition}
Using the BCH formula it is possible to develop a number of identities
including: 
\begin{theorem}
    
\begin{equation}
\hat{D}(\alpha)=e^{-|\alpha|^{2}/2}e^{\alpha\hat{a}^{\dagger}}e^{-\alpha^{*}\hat{a}}
\end{equation}
\end{theorem}
We know from Weyl's transformation law that there is a representation
of the Wigner function that corresponds to this algebra. The authors
in \cite{potovcek2015exponential} showed that one way to realise it is to use
the expression: 
\begin{equation}
\hat{\Delta}(x,p)=\dfrac{1}{(2\pi)^{2}}\int_{-\infty}^{+\infty}\int_{-\infty}^{+\infty}e^{i(px'-xp')}\hat{D}(x,p)dx'dp'
\end{equation}
They also proved the similarity type relation: 
\begin{equation}
\hat{D}(x,p)\hat{\Pi}\hat{D}^{\dagger}(x,p)=\hat{D}(x,p)\hat{D}(x,p)\hat{\Pi}
\end{equation}
\begin{equation}
=\hat{D}(2x,2p)\hat{\Pi}=\hat{\Pi}\hat{D}^{\dagger}(2x,2p)
\end{equation}
where the operator $\hat{\Pi}$ is the parity operator which acts
on the state via 
\begin{equation}
\pi\hat{\Delta}(x,p)=\hat{D}(x,p)\hat{\Pi}\hat{D}^{\dagger}(x,p)
\end{equation}
Some relevant formula are the completeness relationship: 
\begin{equation}
\int_{-\infty}^{+\infty}\int_{-\infty}^{+\infty}\hat{D}(x,p)dx'dp'=\mathbf{1}
\end{equation}
Initial conditions are specified through $\pi\hat{\Delta}(0,0)=\hat{\Pi}$.
The action on the wave function is: 
\begin{theorem}
    
\begin{equation}
[\hat{\Delta}(x,p)\psi](\chi)=\dfrac{1}{\pi}e^{2ip(\chi-x)}\psi(2x-\chi)
\end{equation}
\end{theorem}
Q.E.D.
\end{proof}
These types of formulae are sufficient to determine the Wigner function
on the hyperboloid. We seek a finite matrix version of this above
structure, especially the relationship between the parity and displacement
operator.

\subsection{Squeeze Transform}

The other related transform on this space is given by a squeeze operator,
which can be seen as the Wigner-Weyl transform of a parametric down
conversion.
\begin{proposition}
    The squeeze operator performs a similar role to the displacement operator. Braiding relationships and isomorphisms show that this is intimately related to the parity operator in the phase space. This operator plays a key role in determining the structure of the algebra associated to this representation of phase space quantum mechanics.
\end{proposition}
\begin{proposition}
    It is possible to generate the SU(1,1) algebra by using relationships between creation and annihilation operators.
\end{proposition}
\begin{corollary}
    It is possible to define the basic relationships we desire for a matrix Wigner function through the laws of the parity operator.
\end{corollary}
\begin{proof}
    
 The squeeze operator is a quadratic in the creation and
annihilation operators. 
\begin{definition}
    
\begin{equation}
\hat{S}(z)=\exp\left(\dfrac{1}{2}\left(z^{*}\hat{a}^{2}-z\hat{a}^{\dagger2}\right)\right)
\end{equation}
\begin{equation}
\hat{D}(\alpha)\hat{S}(z)=\hat{S}(z)\left(\hat{S}^{\dagger}(z)\hat{D}(\alpha)\hat{S}(z)\right)=\hat{S}(z)\hat{D}(\gamma)
\end{equation}
\end{definition}
There is another braiding relationship between the squeeze and displacement
operators: 
\begin{lemma}
    
\begin{equation}
\gamma=\alpha\cosh r+\alpha^{*}e^{i\theta}\sinh r
\end{equation}
\end{lemma}
It is not too hard to show that the creation and annhilation operators
act as rotations in the hyperbolic space: 
\begin{theorem}
    
\begin{equation}
\hat{S}^{\dagger}(z)\hat{a}\hat{S}(z)=\hat{a}\cosh r-\hat{a}^{\dagger}e^{i\theta}\sinh r,z=re^{i\theta}
\end{equation}
\end{theorem}
Writing now the displacement operator in phase space coordinates,
we find the Glauber-Sudarshan \cite{glauber1955time} representation:
\begin{equation}
\hat{D}(\alpha)=\exp\left(i(p\hat{x}-x\hat{p})\right)
\end{equation}
Using the same type of substitution, the squeeze operator takes the
form: 
\begin{equation}
\hat{S}(\alpha)=\exp\left(\dfrac{i}{2}\left(\dfrac{p}{\sqrt{2}}(\hat{x}^{2}-\hat{p}^{2})-\dfrac{x}{\sqrt{2}}\left\{ \hat{x},\hat{p}\right\} \right)\right)
\end{equation}
which is of the form of a Wigner-Weyl transform: 
\begin{equation}
\hat{S}(\alpha)=\exp\left(\dfrac{i}{2}\left(p\hat{\mathcal{P}}-x\hat{\mathcal{X}}\right)\right)
\end{equation}
We can use this form of the transform to derive the symmetry properties
of the system: 
\begin{equation}
\left[\begin{array}{c}
\hat{a}(\alpha)\\
\hat{a}^{\dagger}(\alpha)
\end{array}\right]=\left[\begin{array}{cc}
\cosh r & -e^{i\theta}\sinh r\\
-e^{-i\theta}\sinh r & \cosh r
\end{array}\right]\left[\begin{array}{c}
\hat{a}\\
\hat{a}^{\dagger}
\end{array}\right]
\end{equation}
\begin{equation}
\hat{U}(\alpha)\mathbf{b}=\mathbf{b}'
\end{equation}
\begin{theorem}
    
In terms of the phase space coordinates, the transformation takes
the form: 
\begin{equation}
\hat{U}(\alpha)=\left[\begin{array}{cc}
\cosh r-\cos\theta\sinh r & \sin\theta\sinh r\\
\sin\theta\sinh r & \cosh r+\cos\theta\sinh r
\end{array}\right]
\end{equation}
\begin{equation}
\det\hat{U}(\alpha)=1,\hat{U}^{-1}(r,\theta)=\hat{U}(-r,\theta)
\end{equation}
For the angle chosen to be $\theta=0$ we have action upon the phase
space coordinates in the form: 
\[
\hat{U}(r)=\left[\begin{array}{cc}
\cosh r-\sinh r & 0\\
0 & \cosh r+\sinh r
\end{array}\right]=\left[\begin{array}{cc}
e^{-r} & 0\\
0 & e^{r}
\end{array}\right]
\]
\begin{equation}
\left[\begin{array}{c}
\hat{x}'\\
\hat{p}'
\end{array}\right]=\left[\begin{array}{c}
e^{-r}\hat{x}\\
e^{r}\hat{p}
\end{array}\right]
\end{equation}
\end{theorem}
From this perspective, we can see how a squeeze in one parameter results
in the broadening of the other, which is an aspect of the Heisenberg
uncertainty principle. We seek a more direct way to access the implied
symmetry of the Wigner function. Using the following results from
quantum optics \cite{potovcek2015exponential}, the Wigner operator is given by:
\begin{definition}
    
\begin{equation}
\hat{w}(\alpha)=\dfrac{1}{\pi^{2}}\int_{\mathbb{C}}\exp\left(\alpha\beta^{*}-\beta\alpha^{*}\right)\hat{D}(\beta)d\beta
\end{equation}
\end{definition}
The Wigner operator follows the displacement law of quantum optics
via isometric transformation, viz. 
\begin{definition}
    
\begin{equation}
\hat{w}(\alpha)=\hat{D}(\alpha)\hat{w}(0)\hat{D}^{\dagger}(\alpha)
\end{equation}
\end{definition}

Similarly to the original postulates, we have the Wigner operator
is Hermitian, via $\hat{w}(\alpha)=\hat{w}^{\dagger}(\alpha)$. Identically
to the calculation from quantum optics \cite{potovcek2015exponential}, the initial
condition of the Wigner operator is the parity: 
\begin{equation}
\hat{w}(0)=\int_{\mathbb{C}}\hat{D}(\beta)d\beta=2\hat{P}
\end{equation}
where the parity is $\hat{P}=\left(-1\right)^{\hat{a}^{\dagger}\hat{a}}=(-1)^{\hat{N}}$.
The connection between the Wigner distribution and this operator is
given by the tracial relation: 
\begin{equation}
W_{\hat{A}}(\alpha)=\mathbf{Tr}\left[\hat{A}\hat{w}(\alpha)\right]
\end{equation}
The inverse to this map is: 
\begin{equation}
\hat{A}=\dfrac{1}{2\pi^{2}}\int_{\mathbb{C}}\hat{w}(\alpha)W_{\hat{A}}(\alpha)d\alpha
\end{equation}
For the density matrix, we can associate the Wigner function via:
\begin{theorem}
    
\begin{equation}
W_{\hat{\rho}}(\alpha)=\mathbf{Tr}\left[\hat{\rho}\hat{w}(\alpha)\right]
\end{equation}
\end{theorem}
We have that $\hat{\rho}^{\dagger}=\hat{\rho}$ so $W_{\hat{\rho}}^{*}(\alpha)=W_{\hat{\rho}}(\alpha)$,
i.e. is a real function. The displaced density matrix can be written:
\begin{equation}
\hat{\rho}'=\hat{D}(\alpha')\hat{\rho}\hat{D}^{\dagger}(\alpha')
\end{equation}
If we substitute this into the formula for the Wigner function and
rearrange, we find: 
\begin{theorem}
    
\begin{equation}
W_{\hat{\rho}'}(\alpha)=W_{\hat{\rho}}(\alpha-\alpha')
\end{equation}

\end{theorem}
\end{proof}
This is all the basic structure we need to identify a Wigner function.
We can see how the different parts of this operator work together.
In particular the relationship between the trace and the phase space
function is of interest. The parity operator obviously enters in a
fundamental way and defines the symmetries of this operator. We shall
now show how finite matrices enter as a complementary representation
to the continuous systems we have considered so far.

\section{Pseudounitary Matrices and SU(1,1)}

We shall now demonstrate one simple way in which a finite version
of the Wigner distribution can be found. This is achieved by using
the pseudounitary counterparts of SU(2), i.e. SU(1,1). The most basic
representation of this is given by 2x2 matrices; the idea is that
if we can find ways in which to represent relations using small matrices,
we can generalise to any other representation. 

\begin{definition}
    By examining the structure
of the parametric algebra, and defining the group generators as $\hat{K}_{+}=\dfrac{1}{2}\hat{a}^{\dagger2},\hat{K}_{-}=\dfrac{1}{2}\hat{a}^{2}$,
it is not too hard to show that the boson algebra goes over into SU(1,1):
\begin{equation}
\left[\hat{K}_{-},\hat{K}_{+}\right]=\dfrac{1}{4}\left[\hat{a}^{2},\hat{a}^{\dagger2}\right]=\dfrac{1}{2}\left(1+2\hat{a}^{\dagger}\hat{a}\right)=\hat{K}_{0}
\end{equation}
\begin{equation}
\left[\hat{K}_{0},\hat{K}_{\pm}\right]=\pm2\hat{K}_{\pm}
\end{equation}
\end{definition}
Recognising that this is the SU(1,1) algebra, we can use the Pauli
matrix representation: 
\begin{definition}
    
\begin{equation}
\hat{k}_{0}=\hat{\sigma}_{z}=\left[\begin{array}{cc}
1 & 0\\
0 & -1
\end{array}\right],\hat{k}_{+}=i\hat{\sigma}_{+}=\left[\begin{array}{cc}
0 & i\\
0 & 0
\end{array}\right]
\end{equation}
\begin{equation}
\hat{k}_{-}=i\hat{\sigma}_{-}=\left[\begin{array}{cc}
0 & 0\\
i & 0
\end{array}\right]
\end{equation}
\end{definition}
\begin{proposition}
    Using the finite representation of the hyperbolic algebra defined through the parabolic operators, the Wigner function can be written as an isomorphic, unitary transformation of the parity operator. 
\end{proposition}
\begin{proposition}
    The parity laws can be solved for a finite matrix representation of the displacement operator.
\end{proposition}
\begin{proof}
    
Applying decomposition formulae to this representation of the fundamental
algebra, it is possible to show the following:
\begin{lemma}
    
\begin{equation}
\hat{S}(\xi)=\exp\left(r\hat{Q}(\theta)\hat{J}\hat{Q}^{\dagger}(\theta)\right)=\hat{Q}(\theta)\exp\left(r\hat{J}\right)\hat{Q}^{\dagger}(\theta)
\end{equation}
\begin{equation}
=\left[\begin{array}{cc}
\cosh r & ie^{i\theta}\sinh r\\
-ie^{-i\theta}\sinh r & \cosh r
\end{array}\right]
\end{equation}
\begin{equation}
=\exp\left(-e^{-i\theta}\tanh r\hat{k}_{-}\right)\exp\left(\ln\cosh r\hat{k}_{0}\right)\exp\left(e^{i\theta}\tanh r\hat{k}_{+}\right)
\end{equation}
\end{lemma}
This expression is equivalent to the Gauss LDU decomposition and gives
the fundamental relationship for any representation of this hyperbolic
algebra. Seyfarth et. al \cite{seyfarth2020wigner} have established the following
relationships using the method of coherent states. 

\begin{theorem}
The parity inversion theorem: 
\begin{equation}
\hat{D}(\alpha)\hat{P}(\alpha)\hat{D}^{\dagger}(\alpha)=\hat{P}(\alpha)\hat{D}(-2\alpha)=\hat{D}(2\alpha)\hat{P}(\alpha)=\hat{w}(\alpha)
\end{equation}
\end{theorem}
which we recognise as the property derived in the previous sections.

\begin{theorem}
    The formula for the Wigner operator in terms of displacement: 
\begin{equation}
\hat{w}(\alpha,\beta)=2^{2}(-1)^{\hat{a}^{\dagger}\hat{a}+\hat{b}^{\dagger}\hat{b}}\hat{D}(-2\alpha)\hat{D}^{\dagger}(-2\beta)
\end{equation}
\end{theorem}

\begin{theorem}
   The composition formula: 
   \begin{equation}
\hat{w}(\alpha,\beta)=\hat{w}(\alpha)\hat{w}(\beta)
\end{equation}
\end{theorem} 
\begin{theorem}
    
The Wigner operator in terms of displacement of the parity: 
\begin{equation}
\hat{w}(\alpha)=\hat{D}(\alpha)\left[2(-1)^{\hat{a}^{\dagger}\hat{a}}\right]\hat{D}^{\dagger}(\alpha)
\end{equation}
\end{theorem}
We wish to establish a matrix theory of the Wigner operator independent
of the mechanics of coherent states. There are advantages to doing
this, as a concrete matrix representation avoids many of the pitfalls
involved in creation-annihilation operator calculus. If we can specify
the matrix form of the parity operator, we can work the rest out.
\begin{theorem}
    
\begin{equation}
\hat{P}(\Phi)=\exp\left(i\dfrac{\Phi}{2}\hat{k}_{0}\right)=\left[\begin{array}{cc}
e^{i\Phi/2} & 0\\
0 & e^{-i\Phi/2}
\end{array}\right]
\end{equation}
\end{theorem}
If we use this form of the parity operator, any changes due to displacement
will be absorbed into the constant in the argument. This is, intuitively,
the formula for a displaced version of parity; we can show this by
writing out the displacement via:
\begin{theorem}
    
\begin{equation}
\hat{D}(\xi)=\left[\begin{array}{cc}
\cos r & e^{-i\theta}\sin r\\
-e^{i\theta}\sin r & \cos r
\end{array}\right]=\left[\begin{array}{cc}
\alpha & \beta\\
-\beta^{*} & \alpha
\end{array}\right]
\end{equation}
\end{theorem}
Using basic matrix multiplication, it is not too difficult to show
the identity: 
\[
\hat{D}(\xi)\hat{P}(\Phi)\hat{D}^{\dagger}(\xi)
\]
\begin{equation}
=\left[\begin{array}{cc}
e^{i\Phi/2}\cos^{2}r+e^{-i\Phi/2}\sin^{2}r & -i\sin\dfrac{\Phi}{2}\sin2re^{-i\theta}\\
-i\sin\dfrac{\Phi}{2}\sin2re^{i\theta} & e^{-i\Phi/2}\cos^{2}r+e^{i\Phi/2}\sin^{2}r
\end{array}\right]
\end{equation}
and also the parity inversion theorem: 
\begin{theorem}
    
\begin{equation}
\hat{P}(\Phi)\hat{D}(-2\xi)=\hat{P}(\Phi)\hat{D}^{\dagger2}(\xi)
\end{equation}
\end{theorem}
By calculating the other side of the theorem, we can fix the value
of the constant in the parity operator. In this representation, the
matrices are unitary. The other part of the parity inversion theorem
reads as: 
\begin{equation}
\hat{D}^{2}(\xi)\hat{P}(\Phi)=\left[\begin{array}{cc}
e^{i\Phi/2}\cos2r & \sin2re^{-i(\Phi/2+\theta)}\\
-\sin2re^{i(\Phi/2+\theta)} & e^{-i\Phi/2}\cos2r
\end{array}\right]
\end{equation}
So if we choose the constant parameter to be $\Phi=\pi$ we can solve
the expression: 
\begin{equation}
\hat{D}(\xi)\hat{P}(\pi)\hat{D}^{\dagger}(\xi)=\hat{D}^{2}(\xi)\hat{P}(\pi)=\hat{P}(\pi)\hat{D}^{\dagger2}(-\xi)
\end{equation}
The parity operator is then given by: 
\begin{lemma}
    
\begin{equation}
\hat{w}(0)=\hat{P}(\pi)=\left[\begin{array}{cc}
i & 0\\
0 & -i
\end{array}\right]
\end{equation}
which can also be realised by using the resolution of identity: 
\begin{equation}
\hat{w}(0)=\sum_{k}(-1)^{k}\left|k_{i}\right\rangle \left\langle k_{i}\right|
\end{equation}
\end{lemma}
If we look at the group homomorphism of the Wigner operator, we can
write: 
\begin{theorem}
    
\begin{equation}
W_{\hat{\rho}}(\zeta')=\mathbf{Tr}\left(\hat{\rho}\hat{T}\hat{w}(\zeta)\hat{T}^{\dagger}\right)
\end{equation}
\begin{equation}
=\mathbf{Tr}\left(\hat{\rho}\hat{w}(g^{-1}\zeta)\right)=W_{\hat{\rho}}(g^{-1}\zeta)
\end{equation}
\end{theorem}
Q.E.D.
\end{proof}
\begin{proposition}
    The squeeze operator can be solved in finite matrix form using the insights of the previous proof.
\end{proposition}
\begin{proposition}
    This is equivalent, on SU(1,1), to the Fourier transform of the parity operator.
\end{proposition}
\begin{corollary}
    This can be written in a form which resembles the characteristic function of a random variable, with the random variable replaced by a unitary matrix.
\end{corollary}
\begin{proof}
    
The most general operator in SU(1,1) is given by: 
\begin{definition}
    
\begin{equation}
\hat{T}=\hat{S}(\xi)e^{i\Phi\hat{K}_{0}}\hat{S}^{-1}(\xi)
\end{equation}
\end{definition}
Evaluating, we find the matrix representation for the squeeze operator
in the pseudounitary case: 
\begin{lemma}
    
\begin{equation}
\hat{S}(\xi)=\left[\begin{array}{cc}
\cosh\dfrac{\tau}{2} & -ie^{i\chi}\sinh\dfrac{\tau}{2}\\
ie^{-i\chi}\sinh\dfrac{\tau}{2} & \cosh\dfrac{\tau}{2}
\end{array}\right]
\end{equation}
\end{lemma}
The rest of the parts of the group isomorphism are easily calculated:
\begin{equation}
e^{i\Phi\hat{K}_{0}}=\left[\begin{array}{cc}
e^{i\Phi} & 0\\
0 & e^{-i\Phi}
\end{array}\right]
\end{equation}
\begin{equation}
\hat{T}(g)=\left[\begin{array}{cc}
\cos\Phi+i\sin\Phi\cosh\tau & -e^{i\chi}\sin\Phi\sinh\tau\\
-e^{-i\chi}\sin\Phi\sinh\tau & \cos\Phi-i\sin\Phi\cosh\tau
\end{array}\right]
\end{equation}
where we note that $\det\hat{T}=1$, so this is a hyperbolic rotation.
The inverse matrix is given by the element:
\begin{equation}
\hat{T}(g^{-1})=\left[\begin{array}{cc}
\cos\Phi-i\sin\Phi\cosh\tau & e^{i\chi}\sin\Phi\sinh\tau\\
e^{-i\chi}\sin\Phi\sinh\tau & \cos\Phi+i\sin\Phi\cosh\tau
\end{array}\right]
\end{equation}
We shall now show a simple alternative way to access these formulae.
If we examine the trace formula, we note that following Seyfarth we
may write the expectation value \cite{seyfarth2020wigner}:
\begin{theorem}
    
\begin{equation}
W_{\hat{\rho}}(\xi)=\mathbf{Tr}[\hat{\rho}\hat{w}(\xi)]=\left\langle (-1)^{\mu}e^{2i\mu}\right\rangle 
\end{equation}
\end{theorem}
If we examine the physical meaning of this expression, we notice that
we can understand the Wigner distribution as the Fourier transform
of the parity operator. The matrix form of this equation will then
be given by: 
\begin{theorem}
    
\begin{equation}
\hat{w}(\xi)=\exp\left(i\Phi\hat{U}(\xi)\right)
\end{equation}
\end{theorem}
where $\hat{U}$ is the time evolution operator of the system. We
can evaluate the time evolution operator using the expansion over
the generators of SU(1,1): 
\begin{equation}
\hat{U}(\xi)=\hat{K}_{0}\cosh\tau+e^{i\chi}\sinh\tau\hat{K}_{+}+e^{-i\chi}\sinh\tau\hat{K}_{-}
\end{equation}
Using the small-$\hat{k}$ representation over the 2x2 matrices this
is evaluated as the matrix:
\begin{equation}
\hat{U}(\xi)=\left[\begin{array}{cc}
\cosh\tau & -ie^{i\chi}\sinh\tau\\
-ie^{-i\chi}\sinh\tau & -\cosh\tau
\end{array}\right]
\end{equation}

\begin{theorem}
  The Wigner operator in this representation can be understood as a
matrix form of the characteristic function:  
\begin{equation}
\hat{w}(\xi)=\exp\left(i\Phi\hat{U}(\xi)\right)
\end{equation}
\begin{equation}
=\left[\begin{array}{cc}
\cos\Phi+i\sin\Phi\cosh\tau & e^{i\chi}\sin\Phi\sinh\tau\\
e^{-i\chi}\sin\Phi\sinh\tau & \cos\Phi-i\sin\Phi\cosh\tau
\end{array}\right]
\end{equation}

\end{theorem}
which matches the formula obtained by using the group decomposition
law. Q.E.D.
\end{proof}

\begin{remark}
    
As far as the author is aware, this is a new result and has not
been reported in the literature. It is interesting to consider the
natural generalisation of this idiom to other unitary operators. It
is not known which groups will give ``nice'' matrix representations,
but by using the obvious relationships between SU(2)\textasciitilde SO(3)
and SU(1,1)\textasciitilde SO(2,1), it is not too hard to see that
at least for these groups there is a possibility of extension. The
general theory of these groups and matrix decompositions is beyond
the scope of this paper, but by using the example of SU(3), it is
possible to illustrate the complications in finding explicit Wigner
matrices on higher dimensions. We can see immediately how to exploit
this for other systems. We have immediately, the following:
\end{remark} 
\begin{corollary}
    The Wigner function associated to the quantum brachistochrone and time optimal control problem may be evaluated through the use of the matrix characteristic function.
\end{corollary}
\begin{corollary}
    The squeeze operator defines a Hermitian, unitary transformation and in this state space the Wigner function evolves unitarily according to the von Neumann equation.
\end{corollary}
\begin{proof}
    We begin with the basic expression for the matrix characteristic function:
\begin{definition}
    \begin{equation}
\hat{w}(\xi)=\exp\left(i\Phi\hat{U}(\xi)\right)
\end{equation}    
\end{definition}

If we have access to a formula for the time evolution operator, this
gives us a way to define the matrix Wigner operator consistently.
From the theory of the quantum brachistochrone \cite{morrison2012time}, see
e.g. also the theory of group representations in \cite{vilenkin1978special},
we have many formulae for the time evolution operator, e.g. 
\begin{lemma}
    
\begin{equation}
\hat{U}_{3}=e^{i\phi}\left[\begin{array}{cc}
\cos\phi & i\sin\phi\\
i\sin\phi & \cos\phi
\end{array}\right]
\end{equation}
\end{lemma}
The unitary matrix we get from the quantum brachistochrone is not
correct for the Wigner transform. We have $\det\hat{U}=1$. Here we
are required to use a unitary with the relation to the parity operator,
where $\det\hat{U}=-1$. One such operator is given by the matrix:

\begin{theorem}
    
\begin{equation}
\hat{U}(\xi)=\left[\begin{array}{cc}
\cos\tau & -e^{i\chi}\sin\tau\\
-e^{-i\chi}\sin\tau & -\cos\tau
\end{array}\right]
\end{equation}
\end{theorem}
which is the real, unitary form of the transformed parity operator.
We intend to apply the same logic as before to write down the Wigner
operator for this state, and derive the Wigner distribution. Evaluating
the matrix exponential as before, we find: 
\begin{theorem}
    
\begin{equation}
\hat{w}(\xi)=\left[\begin{array}{cc}
\cos\Phi+i\sin\Phi\cos\tau & -i\sin\tau\sin\Phi e^{i\chi}\\
-i\sin\tau\sin\Phi e^{-i\chi} & \cos\Phi-i\sin\Phi\cos\tau
\end{array}\right]
\end{equation}
\end{theorem}
We wish to extend the calculation further, by using the density matrix
we can extract the Wigner function: 
\begin{equation}
W_{\hat{\rho}}(\alpha)=\mathbf{Tr}\left[\hat{\rho}\hat{w}(\alpha)\right]
\end{equation}
To evaluate the density matrix for this Wigner operator, we may utilise
the decomposition over the diagonal eigenstates of the squeeze operator,
finding:
\begin{equation}
\hat{S}(\xi)=\exp\left(-ir\hat{Q}(\theta)\hat{J}\hat{Q}^{\dagger}(\theta)\right)=\hat{Q}(\theta)\exp\left(-ir\hat{J}\right)\hat{Q}^{\dagger}(\theta)
\end{equation}
Writing out the decomposition of the squeeze operator: 
\begin{equation}
\hat{S}(\xi)=\dfrac{1}{2}\left[\begin{array}{cc}
ie^{i\theta} & -ie^{i\theta}\\
1 & 1
\end{array}\right]\left[\begin{array}{cc}
e^{ir} & 0\\
0 & e^{-ir}
\end{array}\right]\left[\begin{array}{cc}
-ie^{-i\theta} & 1\\
ie^{-i\theta} & 1
\end{array}\right]
\end{equation}
From the matrix of eigenstates, we can reconstruct the density matrix:
\begin{theorem}
    
\begin{equation}
\hat{Q}(\theta)=\dfrac{1}{\sqrt{2}}\left[\begin{array}{cc}
ie^{i\theta} & -ie^{i\theta}\\
1 & 1
\end{array}\right]=\left[\left|+\right\rangle ,\left|-\right\rangle \right]
\end{equation}
\begin{equation}
\tilde{H}=r\left[\begin{array}{cc}
0 & e^{+i\theta}\\
-e^{-i\theta} & 0
\end{array}\right],\tilde{H}\left|\pm\right\rangle =\pm r\left|\pm\right\rangle 
\end{equation}
\begin{equation}
i\dfrac{\partial\hat{\rho}}{\partial t}=\left[\tilde{H},\hat{\rho}\right]
\end{equation}
\end{theorem}
From the formula for the density matrix we can derive the result we
need: 
\begin{equation}
\hat{\rho}(t)=\hat{U}(t,0)\hat{\rho}(0)\hat{U}^{\dagger}(t,0)
\end{equation}
\begin{equation}
\hat{U}(t,0)=\hat{W}(t)\hat{W}^{\dagger}(0),\hat{U}^{\dagger}(t,0)=\hat{W}(0)\hat{W}^{\dagger}(t)
\end{equation}
Using the expression for the transformed density, we find the isomorphism
in terms of the solution matrix: 
\begin{theorem}
    
\begin{equation}
\hat{\rho}(t)=\hat{W}(t)\hat{W}^{\dagger}(0)\hat{\rho}(0)\hat{W}(0)\hat{W}^{\dagger}(t)=\hat{W}(t)\hat{\rho}_{w}(0)\hat{W}^{\dagger}(t)
\end{equation}
\end{theorem}
Since we know that the solution matrix follows standard dynamics,
it is simple to show that this expression solves the von Neumann equation
Evaluating the derivative: 
\begin{corollary}
    
\begin{equation}
i\dfrac{\partial\hat{\rho}}{\partial t}=i\dfrac{\partial}{\partial t}\left(\hat{W}(t)\hat{\rho}_{w}(0)\hat{W}^{\dagger}(t)\right)
\end{equation}
\begin{equation}
=\left[\tilde{H},\hat{\rho}\right]
\end{equation}
\end{corollary}
Q.E.D.
\end{proof}
\begin{remark}
    
We can therefore use this expression for the density matrix when calculating
the Wigner function. Choosing an initial state under a given Hamiltonian
is enough to find the density matrix in the simplified situation of
a pure state. As shall be shown in the next section, it is important
to make a distinction between pure and mixed states, as even in this
simplified example of a two state non-commutative system their behaviour
is markedly different. 

\end{remark}
\section{Theory of the Projection Matrix}

The remaining topic of interest relates to the integration theory
of groups. We state the following:
\begin{proposition}
    The projective theory is important to the formulation of the representation theory of the matrix Wigner theory. 
\end{proposition}
\begin{proposition}
    The evolution of the matrix Wigner function is more delicate than the Bloch picture of the density matrix. In particular, pure and mixed states behave very differently.
\end{proposition}
\begin{corollary}
The projective theory of the matrix Wigner function results in a type of Hadamard transformation which may be of interest in quantum mechanics.
\end{corollary}
\begin{proof}
    
 Note the projection formula:
 \begin{definition}
     
\begin{equation}
\hat{\rho}=\dfrac{1}{A}\int_{\Omega}d\mu(\Omega).\hat{w}(\Omega)\mathrm{Tr}\left[\hat{\rho}\hat{w}(\Omega)\right]
\end{equation}
 \end{definition}
We can also analyse the Wigner function using properties of the projection
operator. From basic quantum mechanics, and assuming the system is
unitary SU(2), we can write:
\begin{example}
    
\begin{equation}
\hat{\rho}(t)=\dfrac{1}{2}\left(\hat{\mathbf{1}}+\vec{n}(t)\cdot\mathbf{\sigma}\right)=\left[\begin{array}{cc}
r_{1} & 0\\
0 & r_{2}
\end{array}\right]
\end{equation}
\end{example}
in the simplest instance. The Wigner operator is then evaluated using
the Euler decomposition: 
\begin{equation}
\hat{w}(\theta)=\dfrac{1}{2}\exp\left(i\dfrac{\theta}{2}\hat{\sigma}_{y}\right)\left[\begin{array}{cc}
1+\sqrt{a} & 0\\
0 & 1-\sqrt{a}
\end{array}\right]\exp\left(-i\dfrac{\theta}{2}\hat{\sigma}_{y}\right)
\end{equation}
The result for the Wigner operator is then: 
\begin{lemma}
    
\begin{equation}
\hat{w}(\theta)=\dfrac{1}{2}\left[\begin{array}{cc}
\sqrt{a}\cos\theta+1 & -\sqrt{a}\sin\theta\\
-\sqrt{a}\sin\theta & -\sqrt{a}\cos\theta+1
\end{array}\right]
\end{equation}

\end{lemma}

Note that we satisfy the constraint equations: 
\begin{definition}
    
\begin{equation}
\mathrm{Tr}\left[\hat{w}(\theta)\right]=1,\mathrm{Tr}\left[\hat{w}^{2}(\theta)\right]=\dfrac{1}{2}(a+1)
\end{equation}
\end{definition}
The authors in \cite{hillery1984distribution} conduct analysis of the Stratonovich-Weyl
operator. They give the relationship between the dimension of the
space and the Wigner operator as: 
\begin{equation}
\mathrm{Tr}\left[\hat{w}^{2}(\theta)\right]=N=2
\end{equation}
from which we are able to determine $a=3$. Calculating the Wigner
distribution via the trace relationship gives the expression: 
\begin{equation}
W_{\hat{\rho}}(\theta)=\mathbf{Tr}[\hat{\rho}\hat{w}(\theta)]
\end{equation}
we find the Wigner distribution: 
\begin{theorem}
    
\begin{equation}
W_{\hat{\rho}}(\theta)=\mathrm{Tr}\left(\dfrac{1}{2}\left[\begin{array}{cc}
r_{1} & 0\\
0 & r_{2}
\end{array}\right]\left[\begin{array}{cc}
\sqrt{3}\cos\theta+1 & -\sqrt{3}\sin\theta\\
-\sqrt{3}\sin\theta & -\sqrt{3}\cos\theta+1
\end{array}\right]\right)
\end{equation}
\begin{equation}
=\dfrac{\sqrt{3}}{2}(r_{1}-r_{2})\cos\theta+\dfrac{r_{1}+r_{2}}{2}
\end{equation}
\end{theorem}
To close the proof we must show that we can recover the operator via
the integral identity: 
\begin{equation}
\hat{\rho}=\dfrac{1}{A}\int_{\Omega}d\mu(\Omega).\hat{w}(\Omega)\mathrm{Tr}\left[\hat{\rho}\hat{w}(\Omega)\right]
\end{equation}
From our initial density matrix, we know from the projection theorem
that we must have: 
\begin{equation}
\hat{\rho}=\dfrac{1}{2}(\hat{\mathbf{1}}+r\hat{\sigma}_{z})=\dfrac{1}{2}\left[\begin{array}{cc}
1+r & 0\\
0 & 1-r
\end{array}\right]=\left[\begin{array}{cc}
r_{1} & 0\\
0 & r_{2}
\end{array}\right]
\end{equation}
and we finally obtain: 
\begin{theorem}
    
\begin{equation}
\int_{0}^{2\pi}\left[\dfrac{\sqrt{3}}{2}r\cos\theta+\dfrac{1}{2}\right]\hat{w}(\theta)d\theta=\dfrac{\pi}{2}\left[\begin{array}{cc}
1+\dfrac{3}{2}r & 0\\
0 & 1-\dfrac{3}{2}r
\end{array}\right]
\end{equation}
\end{theorem}
Note that this is a mixed state. We have $\hat{\rho}^{2}\leq\hat{\rho}$.
The theory partially fails because of this. We readjust and use the
density of a pure state. A simple formula for a pure state is given
by: 
\begin{definition}
    
\begin{equation}
\hat{\rho}=\dfrac{1}{2}(\hat{\mathbf{1}}+r_{x}\hat{\sigma}_{x}+r_{z}\hat{\sigma}_{z})
\end{equation}
\[
\hat{\rho}=\dfrac{1}{2}\left[\begin{array}{cc}
1+\dfrac{1}{\sqrt{2}} & \dfrac{1}{\sqrt{2}}\\
\dfrac{1}{\sqrt{2}} & 1-\dfrac{1}{\sqrt{2}}
\end{array}\right],\mathrm{Tr}(\hat{\rho})=1,\hat{\rho}^{2}=\hat{\rho}
\]
\end{definition}
Using the same expression for the rotation as before: 
\begin{equation}
\hat{w}(\theta)=e^{i\theta\hat{\sigma}_{y}/2}\hat{\rho}(0)e^{-i\theta\hat{\sigma}_{y}/2}
\end{equation}
\begin{equation}
=\dfrac{1}{2}\left[\begin{array}{cc}
1+\dfrac{\sqrt{2}}{2}\left(\sin\theta+\cos\theta\right) & -\dfrac{\sqrt{2}}{2}\left(\sin\theta-\cos\theta\right)\\
-\dfrac{\sqrt{2}}{2}\left(\sin\theta-\cos\theta\right) & 1-\dfrac{\sqrt{2}}{2}\left(\sin\theta+\cos\theta\right)
\end{array}\right]
\end{equation}
Evaluating the trace to extract the Wigner distribution, we find:
\begin{equation}
\mathrm{Tr}[\hat{w}(\theta)]=\mathrm{Tr}[\hat{w}(\theta)^{2}]=1
\end{equation}
\begin{equation}
W_{\rho}(\theta)=\mathrm{Tr}[\hat{w}(\theta)\hat{\rho}(0)]=\dfrac{1}{2}(\cos\theta+1)
\end{equation}
If we then go on to calculate the integral, it is simple to show that
we receive: 
\begin{theorem}
    
\begin{equation}
\int_{0}^{2\pi}W_{\rho}(\theta)\hat{w}(\theta)d\theta=\dfrac{\pi}{2}\left[\begin{array}{cc}
1+\dfrac{1}{\sqrt{2}} & \dfrac{1}{\sqrt{2}}\\
\dfrac{1}{\sqrt{2}} & 1-\dfrac{1}{\sqrt{2}}
\end{array}\right]
\end{equation}
\begin{equation}
=\pi\hat{\rho}(0)
\end{equation}
\end{theorem}
We can see that in this instance, we recover both the normalisation
and the original pure state density matrix. Q.E.D.
\end{proof} 

It has been shown how
to link sets of special functions and the Wigner distribution in the
first part of this work. The second half has focused on matrix models
which can be compared to these continuous results. We can see that
in both cases, there are certain underlying themes of projection.
We can see this reflected in the original expression for the Wigner
distribution: 
\begin{equation}
f_{W}(p,q)=\dfrac{1}{2\pi\hbar}\int_{-\infty}^{+\infty}e^{ips/\hbar}\left\langle q-\dfrac{s}{2}\right|\left.\Psi\right\rangle \left\langle \Psi\right.\left|q+\dfrac{s}{2}\right\rangle ds
\end{equation}
where we can see the intersection between projective theorems of the
density matrix, autocorrelation, and the Fourier transform. This concludes
the calculation in this paper. We shall now discuss some topics in
analysis that relate to the future developments that can be expected
in this area.

\section{Discussion and Conclusions}

We have shown in this paper some interesting relationships between
the construction of the Wigner function via the star-eigenvalue equations,
and the group representations of matrices and parity operators. These
results have consequences for the theory of projective operators.
It is relevant to consider the direction of using these types of projections
to map out spaces by projecting down to cover the different functions
and matrices. The question of the correct types of projection to be
used in higher dimensional spaces has not been analysed in this paper,
and can be expected to be different to that of these simple contrived
examples. This calculation has demonstrated that a close inspection
of the spectral theory can result in much understanding of the structure
of different groups. Here we have presented the basic examples of
spherical and pseudospherical systems. These methods, with appropriate
modifications can be shown to be applicable to such varied systems
as the Morse oscillator, modified Bessel function and associated Legendre
functions. This will be covered in future works to appear in the literature. 

It is of interest that the particular sets of special functions uncovered
in this calculation indicate a relationship between the confluent
hypergeometric function, the Laguerre functions and the Hermite polynomials.
Indeed, we have found a number of multiplication theorems that essentially
amount to different types of convolutions in these systems. These
convolutions are related to the kernel; as we have seen, the different
forms of hyperbolic systems have differing eigenfunctions that make
up the spectral solutions. 

It is hoped that by utilising the intersection of matrix representation
theory and special functions that a picture of the different families
can be built up by understanding the limits between different sets
of PDE systems. One way to understand this, as we have demonstrated
in this paper, is through analysis of the projection operator and
transform theory using Wigner functions. This is a powerful technique,
and it is possible that a more general understanding of the theory
of diffusion problems and stochastic processes can be found by using
this perspective.

The matrix perspective in this sense requires more development, as
it is not as clear as it would be desirable. The use of group representation
techniques should enable these simple results to be generalised to
any other system with SU(1,1) symmetry. The groups SL(2,R) and SL(2,C)
form obvious extensions of the method and have relationships to SU(1,1)
transforms \cite{lenz2017mehler}. The important part in which more understanding
is required is the difference between trace class of Wigner matrices
and their determinants. As is easily shown, the mixed state we have
calculated has trace class 2, the pure state trace class 1. Whether
other well defined trace classes of Wigner matrices exist, and the
exact analysis of the ``defect'' between pure and mixed states with
respect to the Wigner function is of crucial importance in understanding
the matrix form of these equations.

\section{Future Directions}

Several questions are opened up through this investigation. The details
of the exact link between the Wigner function and the kernel have
not been explored in this paper. If we take the kernel to be the projection
state:

\begin{equation}
K(x,x';t)=\sum_{n}e^{-E_{n}t}\psi_{n}^{*}(x')\psi_{n}(x)
\end{equation}
\begin{equation}
\psi_{n}(x)=\left\langle n|x\right\rangle 
\end{equation}
\begin{equation}
\psi_{n}^{*}(x')=(\left\langle n|x'\right\rangle )^{*}=\left\langle x'|n\right\rangle 
\end{equation}
\begin{equation}
K(x,x';t)=\sum_{n}e^{-E_{n}t}\left\langle x'|n\right\rangle \left\langle n|x\right\rangle =\left\langle x'\left|\sum_{n}e^{-E_{n}t}\left|n\left\rangle \right\langle n\right|\right|x\right\rangle 
\end{equation}
It is tempting, in light of the formula for the Wigner function, to
hope for a simple relationship such as:
\begin{equation}
f_{W}(p,q;t)=\dfrac{1}{2\pi\hbar}\int_{-\infty}^{+\infty}e^{ips/\hbar}K(q-s/2,q+s/2;t)ds
\end{equation}
It is unknown if such a formula exists. This can be easily recovered
by using
\begin{equation}
\left|\Psi\left\rangle \right\langle \Psi\right|=\sum_{n}e^{-E_{n}t}\left|n\left\rangle \right\langle n\right|
\end{equation}

If such relationships could be developed, it would save a great deal
of calculation, as we could exploit already known formulae for the
kernel to find solutions to the simultaneous differential equations
given through the star equations. Other interesting formulae might
be found by using the expression for the resolvent/Green's function
as given by a similar eigenfunction decomposition.

\section*{Acknowledgment}
This project was supported under ARC Research Excellence Scholarships at the University of Technology, Sydney. The author acknowledges useful discussions and support from Dr. Mark Craddock and Prof. Anthony Dooley.


\bibliographystyle{ptephy}
\bibliography{sample}

%

\vspace{0.2cm}
\noindent


\let\doi\relax


\end{document}